\documentclass[11pt]{article}

\usepackage[margin=1in]{geometry}

\usepackage{verbatim}

\usepackage{rotating}
\usepackage{makecell}
\usepackage{natbib}
\usepackage{booktabs}
\usepackage{multirow}

\usepackage{tcolorbox}

\newcolumntype{C}[1]{>{\centering\arraybackslash}p{#1}}

\usepackage[multiple]{footmisc}

\usepackage{float}
\usepackage{comment}
\usepackage{soul}
\usepackage{wrapfig}

\usepackage{amsmath,amsthm,amssymb,mathtools,bbm}

\usepackage{graphicx}
\usepackage{dsfont}
\usepackage{tikz}
\usetikzlibrary{calc}
\usetikzlibrary{shapes.geometric}
\usepackage{paralist}

\usepackage{wrapfig}

\usepackage{float}
\usepackage{comment}
\usepackage{algorithm}
\usepackage[noend]{algpseudocode}
\usepackage{algpseudocodex}

\usepackage{thmtools}
\usepackage{thm-restate}

\usepackage{caption}
\usepackage[labelformat=simple]{subcaption}

\usepackage{aliascnt}

\usepackage{xcolor}

\usepackage[backref=page]{hyperref}
\hypersetup{colorlinks=true,citecolor=blue,linkcolor=blue}
\hypersetup{colorlinks=true}
\hypersetup{linkcolor=[rgb]{.7,0,0}}
\hypersetup{citecolor=[rgb]{0,.7,0}}
\hypersetup{urlcolor=[rgb]{.7,0,.7}}

\declaretheorem[numberwithin=section]{theorem}
\declaretheorem[sibling=theorem, style=definition]{definition}
\declaretheorem[sibling=theorem]{lemma}

\declaretheorem[sibling=theorem, style=definition]{remark}
\declaretheorem[sibling=theorem]{proposition}

\makeatletter
\patchcmd{\ALG@step}{\addtocounter{ALG@line}{1}}{\refstepcounter{ALG@line}}{}{}
\newcommand{\ALG@lineautorefname}{Line}
\makeatother

\newcommand{\R}{\mathbb{R}}

\DeclareMathOperator*{\supp}{supp}

\renewcommand{\P}[1]{\mathbb{P}\left[{#1}\right]}
\newcommand{\E}[1]{\mathbb{E}\left[{#1}\right]}

\makeatletter
\newcommand\Ps@textstyle[2]{\mathbb{P}_{#1}\left[{#2}\right]}
\newcommand\Es@textstyle[2]{\mathbb{E}_{#1}\left[{#2}\right]}
\newcommand\Ps[2]{%
  \mathchoice %
  {\underset{{#1}}{\mathbb{P}}\left[{#2}\right]}
  {\Ps@textstyle{#1}{#2}}
  {\Ps@textstyle{#1}{#2}}
  {\Ps@textstyle{#1}{#2}}
}
\newcommand\Es[2]{%
  \mathchoice %
  {\underset{{#1}}{\mathbb{E}}\left[\vphantom{\big|}{#2}\right]}
  {\Es@textstyle{#1}{#2}}{\Es@textstyle{#1}{#2}}{\Es@textstyle{#1}{#2}}
}
\makeatother

\usepackage{array}      %
\usepackage{multirow}

\hypersetup{pdfauthor={XYZ, Clayton Thomas}}

\algdef{SE}[DOWHILE]{Do}{doWhile}{\algorithmicdo}[1]{\algorithmicwhile\ #1}%

\usepackage{comment}

\usepackage{sidecap}

\usepackage{csquotes}
\usepackage{colortbl}

\usepackage{tikz}
\usetikzlibrary{automata, positioning, calc, fit, shapes.misc}
\usetikzlibrary{matrix}
\usetikzlibrary{positioning}
\usetikzlibrary{arrows, decorations.pathmorphing}
\usetikzlibrary{decorations.pathreplacing,calligraphy}

\usepackage{tikz}

\newcommand{\urtriangle}{%
  \mathbin{%
    \tikz[baseline=-0.5ex, scale=0.15]{
      \draw (0,1) -- (1,1) -- (1,0) -- cycle;
    }%
  }%
}

\usepackage[USenglish]{babel}
\usepackage[useregional]{datetime2}
\DTMlangsetup[en-US]{showdayofmonth=false}

\begin{document}

\title{
Learning from a Mixture of Information Sources
}
\author{
  Nicole Immorlica\\
  Microsoft Research
  \and
  Brendan Lucier\\
  Microsoft Research
  \and
  Yaroslav Mukhin\\
  Cornell University
  \and
  Clayton Thomas\\
  Rensselaer Polytechnic Institute
  \and
  Ruqing Xu\thanks{Immorlica: nicimm@gmail.com, Lucier: brlucier@microsoft.com, Mukhin: ym545@cornell.edu, Thomas: thomas.clay95@gmail.com, Xu: rx24@cornell.edu (corresponding author). 
This research
was conducted prior to the author’s employment by Amazon. This paper is not sponsored or endorsed
by, or associated with Amazon or any of its subsidiaries or affiliates. The views, opinions and positions
included in this paper are the author’s own and do not reflect the views, opinions and positions of Amazon.
We thank Marco Battaglini, Takuma Habu, Rohit Lamba, Mira Frick, Navin Kartik, Ankur Mani, Krishnamurthy Iyer, Thodoris Lykouris, Yiding Feng, Liang Lyu, Drew Fudenberg, Tuomas Sandholm, and participants at Cornell Microeconomic Theory reading group, Cornell Artificial Intelligence, Policy, and Practice (AIPP), Marketplace Innovation Workshop (MIW), and ESIF Economics and AI+ML Meeting for helpful discussions and comments.}\\
  Amazon
}

\begin{titlepage}
\maketitle
\begin{abstract}
  We often learn from multiple sources that convey information in different ways. 
How informative is it to know the source of a signal, and how is this informativeness shaped by the distribution of sources?
We extend the standard (binary-state, binary-signal) Blackwell experiment model by introducing a commonly known distribution over signaling schemes, representing the distribution of information sources.
We compare learning under two information models: \emph{source-aware}, where decision makers observe a signaling scheme and its realization (e.g., raw reviews, search results), and \emph{source-blind}, where only the signal realization is observed (e.g., aggregate ratings, LLM-generated summaries). 
We show that a mean-preserving spread in the distribution of signaling schemes translates into Blackwell dominance for source-aware decision makers, implying they are ``risk-loving'' in information sources.
In contrast, it has no impact on source-blind decision makers.
When learning from repeated draws of signaling schemes, source-blind decision makers learn more slowly.
However, as long as the average signaling scheme is $\varepsilon$ away from being completely uninformative, source-blind learning can match source-aware learning by using at most $O(1/\varepsilon)$ times more data.

\end{abstract}
\thispagestyle{empty}
\end{titlepage}

\maketitle

\section{Introduction}
\label{sec:intro}

In economic decision-making, individuals and institutions routinely rely on information from diverse sources. Whether a piece of information---such as a product review, a news article, or a medical recommendation---is informative depends not only on its content but critically on who produces it. A satirical news report, for instance, invites different interpretations compared to a reputable news outlet; medical advice from a layperson lacks the authority of a licensed professional; and the relevance of a restaurant review hinges on how the tastes of the reviewer and the reader compare. 
In these examples, knowing the source of a piece of information can be critical to evaluate its content and to inform decision making.

On the other hand, modern technology offers unprecedented opportunities for summarization and processing large amounts of data at low costs.
While aggregating information opens up the possibility of harnessing large datasets, the source of each data point---reflecting its reliability, bias, or relevance---is often obscured in the process. 
In the age of big data, the common wisdom is that this loss of context can be overcome with ever-increasing amounts of data.
This raises fundamental questions about the role of data provenance in learning. How is informativeness shaped by the distribution of sources?
How does the value of knowing a signal's source compare with the ability to process more samples?

Our paper proposes a framework for answering  these questions.
There is an unknown decision-relevant state of the world. A signaling scheme is a Blackwell experiment that maps the state to a probability distribution over finite signal realizations.
We assume access to a commonly known \emph{distribution} over signaling schemes, representing heterogeneity in information sources.
We consider a learner who can observe multiple signal realizations, each generated by a separate signaling scheme independently drawn from this distribution.
We compare two learning environments, representing different amounts of information available to the learner.
In the \emph{source-aware} setting, the learner observes the details of the signal-generating processes (e.g., reading raw reviews or examining the origin of search results) in addition to the realizations of the signals from those sources. 
In the \emph{source-blind} setting, only the signal realizations are observed by the learner (e.g., seeing aggregated ratings or summary statistics) while the nature of the signal-generating processes is obscured.  
We can ask our main question in the language of the model:

\begin{quote}%
\emph{%
    How many more samples does a source-blind learner require to be as informed---and hence as effective at decision-making---as a source-aware learner? %
}
\end{quote}

As an initial foray into answering this question, we instantiate our framework in a simple binary setting.  There is a state of the world $\theta \in \{0,1\}$, initially unknown to the learner.  Each potential signaling scheme generates a binary signal that can be correlated with $\theta$.  We assume ``no fake data,'' meaning that no signal realization is negatively correlated with the state, but we do allow perfectly uninformative signals that are not correlated with $\theta$.  While restrictive, this binary scenario is nevertheless sufficiently rich to illustrate the gap between source-blind and source-aware learning, while also capturing natural scenarios in which a decision-maker must select whether or not to take an action (such as purchasing a product) given a dataset that takes the form of action recommendations.

In our main result, we give an upper bound on the advantage of being source-aware when learning from many different sources drawn from the distribution.  We prove that if the \emph{average} of all signaling schemes is at least an $\varepsilon$-distance (in an appropriate metric) from being completely uninformative, then a source-blind learner needs no more than $O(1/\varepsilon)$ times as much data to be at least as informed as a source-aware learner in the sense of Blackwell dominance.  This means, in particular, that for any decision problem that depends on the state $\theta$, a source-aware decision maker will be outperformed (in expectation) by a source-blind decision maker with access to $O(1/\varepsilon)$ times as many signals.

Our bound depends on the informativeness of an average signaling scheme with respect to the distribution from which they are drawn.  To see why, imagine a distribution in which a very small fraction of signals are fully informative of the state, while the rest are completely uninformative.  A source-aware learner would then learn the state perfectly if even a single informative signal is observed, whereas for a source-blind learner a single informative signal is likely to become lost in the noise.  However, if sufficiently many informative signals are present in the dataset, then the state could be inferred even from aggregate statistics.  Our analysis essentially shows that this type of example is the worst-case scenario for a source-blind learner, and the resulting bound applies to all potential distributions over data sources.

It is important to delineate the scope of our main result, which rests on two restrictions. First, it is limited to binary states because its proof relies on the characterization of large-sample Blackwell dominance by \citet{mu2021blackwell}, which is established only for that case. \citet{farooq2024matrix} have recently extended this characterization to finitely many states, offering a promising path to extending our results beyond two states, though such an extension is beyond the scope of this initial study. Second, the no-fake-data assumption helps to rule out degenerate cases. For instance, if truthful and fake sources each fully reveal the state, a source-aware learner learns the state from a single signal, whereas for a source-blind learner the two can cancel out, making the advantage of knowing the source unbounded. In \autoref{sec:value-multiple-signal}, we discuss such cases and explain why a naive generalization to this domain leads to uninteresting results.

Along the way to establishing our main result, we analyze how changes in the underlying distribution over signal-generating processes affect the rate at which a source-aware or source-blind decision maker becomes informed.
We show that as the distribution over signal-generating processes becomes ``more spread out,'' in the sense of a mean-preserving spread, a source-aware decision maker becomes strictly more informed in the sense of Blackwell informativeness. %
Again, this implies that the source-aware decision-maker gets weakly higher expected utility, under any decision problem, when learning from more heterogeneous data sources.
In contrast, we show that the source-blind decision maker cares only about the mean signal, and hence her level of informativeness is unchanged when the source become more spread out in the same sense as above.  These results, which may be of independent interest, apply to any finite signal and state space (not necessarily binary).

\paragraph{Related Work.}

There is a \emph{vast} literature on learning from signaling schemes, i.e., Blackwell experiments \citep{blackwell1951comparison,blackwell1953equivalent}.
This literature considers a host of important topics and results, including (but not limited to) how different information experiments can feasibly correlate 
\citep[e.g.,][]{brooks2022information, brooks2024comparisons},
how agents learn in social environments \citep[e.g.,][]{banerjee1992simple, huang2024learning},
how costs can be assigned to experiments \citep[e.g.,][]{pomatto2023cost,bloedel2020cost}, 
and how experiments can be ranked \citep[e.g.,][]{moscarini2002law, azrieli2014comment,kosenko2024beyond}.
Our main result in \autoref{sec:over-provisioning} relies on a related concept from \cite{mu2021blackwell}, the dominance ratio, which gives a quantitative comparison between two experiments based on the information content of many repeated independent samples from  the experiments.

There is relatively little work that, like our model, considers learning in a \emph{distribution} over signaling schemes.
\citet{acemoglu2016fragility} consider agents with different subjective priors over signaling schemes.
They consider how agents learn from repeated signal realizations generated by a single signaling scheme drawn once from their prior. In contrast, our main result in \autoref{sec:over-provisioning} consider learning from %
many repeatedly-drawn signaling schemes from a commonly-known distribution of such schemes, so there is no learning about signal-generating processes overtime.
\citet{cripps2008common} considers agents that learn from correlated signals and studies conditions under which joint learning occurs and is common knowledge.  In our scenario, a source-aware learner has only more information than the source-blind learner, and hence joint learning is assured at the rate at which the source-blind agent learns; our focus is on comparing the rates at which they learn.
Our result in \autoref{thm:triangle-spread-to-posterior-spread} concerning the mean-preserving spread of signaling schemes generalizes a part of a result in \citet{whitmeyer2024blackwellmonotoneupdatingrules}, which finds that a Bayesian decision maker weakly prefers randomization over two experiments to their convex combination.
They define this as ``convexity,'' and show that Bayes' rule is the only updating rule satisfying convexity and certain other requirements.

Beyond the context of learning, recent works on communication games 
have incorporated ambiguity into signaling schemes, similar to how we introduce heterogeneous signaling schemes. In ambiguous persuasion \citep{beauchene2019ambiguous,cheng2023ambiguouspersuasionexanteformulation, cheng2024persuasion},
the sender is allowed to use ambiguous communication devices comprising a finite set of signaling schemes, while the receiver is assumed to be ambiguity-averse and to maximize their maxmin expected utility. 
Similarly, \citet{kellner2017modes} studies cheap talk
communication with purely exogenous ambiguity and \citet{kellner2018endogenous} studies when sender can strategically choose to communicate ambiguously.

Other papers study conceptually similar questions to ours in different settings.
Motivated by investigating the ``value of context,'' \citet{iakovlev2024value}
ask whether an agent prefers to be evaluated by a human who observes a smaller but customized set of covariates, or an algorithm which observes a bigger set of non-customizable covariates. 
They find that the expected value of customization vanishes to zero as the number of covariates grows.
In a more classic setting, several papers have considered the value of context when learning from signals generated by individuals connected in a social network. For example, if individuals receive binary signals about an unknown state of the world and repeatedly report their best guess of this state to their neighbors, the ``context'' (i.e., the structure of the network) impacts whether the society can learn the state.  If individuals are Bayesian and correctly incorporate the context in their updating process, learning will be asymptotically efficient for a wide range of dynamics \cite{gale2003bayesian,mossel2015strategic,mossel2016efficient}. However, if individuals simply update using a majority dynamic, then the society might fail to learn the state, and multiple beliefs might persist indefinitely unless the social network and dynamics satisfy certain properties \cite{feldman2014reaching,mossel2014majority}.
Another domain related to the value of context is the problem of learning mixture models, whereby data is known to be generated from multiple sources but the source label (i.e., context) is hidden. Depending on the underlying data distribution from each individual source, learning the mixture model may be substantially more difficult than if the source labels were revealed~\cite{dasgupta1999learning,dasgupta2005learning}.

Our main result, \autoref{thm:overprovisioning}, shows that if the average signal is not too uninformative, then for any decision problem, a source-blind learner with access to only $O(1/\epsilon)$ times more
signals will outperform a source-aware learner. This has a similar flavor to \cite{bulow1994auctions}'s seminal result that in independent private value auctions, an auctioneer can generate more profit by adding one more bidder than by switching to the optimal auction mechanism. Similarly, the seminal paper in algorithmic game theory, \cite{roughgarden2002bad}, shows that a network suffers less total latency with Nash routing and twice the edge capacity than that of the optimal routing. As such, our result joins the common refrain in the Economics and Computation literature that the optimal mechanism can often be bested by a simple mechanism plus a bit more resources.

\paragraph{Paper Organization.}

The rest of the paper proceeds as follows.
We define our model of source-aware and source-blind learning in \autoref{sec:prelims}.
We give background on the main concept needed to express our results---the Blackwell order---in \autoref{sec:prelims-blackwell},
and makes basic observations comparing source-aware and -blind learning in \autoref{sec:model}.
We study source-aware learners in \autoref{sec:triangle-spread}, and prove they are ``risk-loving'' in information sources.
\autoref{sec:binary} introduces a tractable special-case of our general model: that of binary states and signals under a ``no-fake-data'' assumption.
In \autoref{sec:value-of-sources}, we consider this special case and give our main result, which bounds the advantage of knowing the source when learning from many sources.
\autoref{sec:conclusion} concludes. We discuss the challenges in extending the analysis to more general states and signaling schemes in \autoref{sec:value-multiple-signal}.
All proofs in the main text are deferred to \autoref{sec:omitted-proofs}.

\section{Model and Preliminaries}
\label{sec:prelims}

\subsection{Signaling Schemes}
\label{sec:prelims-signals}

Let $\Theta = \{0, 1, \dotsc, k-1\}$ be a finite set of $k \geq 2$ \emph{states}.
For concreteness, we assume that the decision maker initially has a uniform prior belief over $\Theta$; 
since our results all concern Blackwell dominance and other prior-free notions, standard arguments imply that our results hold for any other fixed prior.

A \emph{signaling scheme} is a function $\pi: \Theta \to \Delta(\Omega)$ for some finite realization space $\Omega$. 
We refer to an $\omega \in \Omega$ as a \emph{signal realization}.
We can equivalently view a signaling scheme $\pi$ as a 
set of probability measures $(P^{\pi}_\theta)_{\theta \in \Theta}$
over $\Omega$, where $P^{\pi}_\theta$ is the measure of signals realized when the state is $\theta$. 
We use $p^{\pi}_{\theta\omega} = \pi(\omega \mid \theta)$ for $\omega\in \Omega$ and $\theta\in \Theta$ to denote the specific entries in the probability measure $P^{\pi}_{\theta}$.\footnote{
    The standard notation $\pi(\omega \mid \theta)$ represents $\P{\pi(\theta) = \omega}$, the probability of observing signal realization $\omega$ given that the state is $\theta$.
    Additionally, when no confusion can arise, we write a distribution (e.g., $\pi(\theta)$ or $\tau$) with the same notation as a random variable following that distribution (e.g., we consider events $\pi(\theta)=\omega$ or $\tau = q$). } Whenever it is clear which signaling scheme we refer to, we may omit the superscript and write $P_\theta$ and $p_{\theta\omega}$.
We let $\mathcal{P}$ denote the set of feasible signaling schemes, which we refer to as the \emph{domain}.

Given a signaling scheme $\pi$, let $\tau_{\pi} \in \Delta(\Delta(\Theta))$ denote the distribution over posterior beliefs which arises from the signaling scheme $\pi$. 
Since $|\Theta| = k$, 
we typically represent $\tau_{\pi}$ as an element of $\Delta(S_k)$, where $S_k = \{ (q_0, \dotsc, q_{k-1}) \in [0,1]^k\ |\ \sum_{i=0}^{k-1} q_i = 1 \}$ denotes the $k$-simplex and the event $\tau_\pi = (q_0, \dotsc, q_{k-1})$ represents the posterior belief that the state is $i$ with probability $q_i$.

For an concrete example of a signaling scheme and the induced distribution of posterior beliefs, we refer the reader to \autoref{sec:binary} for a discussion under binary state and signal spaces.

\subsection{Learning and Information Models}

We focus on a decision maker who wishes to learn about the state $\theta \in \Theta$ from one or more signal realizations. 
Informally, these are realized according to signal processes that are themselves repeatedly drawn from a  \emph{distribution} $\Pi$ over the domain of signaling schemes $\mathcal{P}$, i.e., $\Pi \in \Delta(\mathcal{P})$.
The standard Blackwell experiment framework with a single signaling scheme $\pi \in \mathcal{P}$ thus corresponds to the distribution $\Pi$ which is a single point mass at $\pi$.

We compare learning under two information models: source-aware and source-blind learning.

First, we introduce source-aware learning.
Intuitively, we say a decision maker is \emph{source-aware} if she learns how a piece of information is generated, together with the content of the information. 
This learning regime may represent environments where one reads raw reviews on a platform or read sources from a search engine.

Formally, we define the source-aware model as follows.
Nature draws (i) a state 
$\theta \in \Theta$
according to the uniform prior, (ii) a signal scheme $\pi \sim \Pi$, and (iii) a signal realization $\omega \sim P^{\pi}_{\theta}$.
The decision maker learns the tuple $(\pi, \omega)$, i.e., the signal scheme and a realization from it.
This process is equivalently a single signal scheme with an enriched signal space $\mathcal{P} \times \Omega$.
Note that the draw of $\pi$ from the distribution $\Pi$ is independent of the state---only the realization $\omega$ is state-dependent.
We denote this signaling scheme by $A(\Pi)$ where $A$ is a mnemonic for ``aware.''

Second, we introduce source-blind learning.
Intuitively, we say the decision maker is \emph{source-blind} if she sees a piece of information without learning its source. 
This learning regime may represent environments where one reads aggregated ratings or generated summaries where information about sources are often suppressed.

Formally, in source-blind learning, nature draws a random $\theta, \pi,$ and $\omega$ exactly as before,
but the decision maker only sees $\omega$ and not $\pi$. 
Thus, this constitutes its own signaling scheme with signal space $\Omega$.
We refer to this signaling scheme as $B(\Pi)$ where $B$ is a mnemonic for ``blind.''

\section{Background: The Blackwell Order}
\label{sec:prelims-blackwell}

We are interested in how useful different signaling schemes are to decision makers.
To that end, in this section we review a classic tool for comparing different information sources.

The canonical  way to capture and compare the usefulness of signaling schemes is via the Blackwell informativeness ordering \citep{blackwell1953equivalent}, which is a partial order $\succeq$ over signaling schemes such that $P \succeq Q$ if and only if \emph{every} decision maker prefers $P$ over $Q$.
We use several classical equivalent definitions for the Blackwell order, particularly using the concepts of garblings of signaling schemes and mean-preserving spreads of posterior distributions.

\begin{definition}[Garblings of signaling schemes]
  \label{def:garble}
  Let $P : \Theta \to \Delta(\Omega)$ and $Q : \Theta \to \Delta(\Xi)$ be two signaling schemes with state space $\Theta$ and realization spaces $\Omega$ and $\Xi$.
  We say that a \emph{garbling function} from $P$ to $Q$ is a function $g: \Omega \to \Delta(\Xi)$ such that for all $\theta\in\Theta$, if we sample $\omega\sim P(\theta)$ and $\xi_g \sim g(\omega)$, then $\xi_g$ is equal in distribution to $Q(\theta)$.
  We say $Q$ is a \emph{garbling} of $P$ if there exists a garbling function from $P$ to $Q$.
\end{definition}

\begin{definition}[Mean-preserving spread in $\R^m$]
  \label{def:mps}
  For random variables $X_1$ and $X_2$ in $\Delta(\R^m)$, we say that a \emph{spread function} from $X_1$ to $X_2$ is a function $s: \R^m \to \Delta(\R^m)$ such that:
  \begin{enumerate}[(1)]
      \item For all $t$ in the support of $X_1$, we have $\E{s(t)} = t$. %
      \item %
      If we draw $z \sim X_1$ and then $y \sim s(z)$, then $y$ is equal in distribution to $X_2$.
  \end{enumerate}
  We say $X_2$ is a \emph{mean-preserving spread} of $X_1$ if there exists a spread function from $X_1$ to $X_2$.\footnote{
    Equivalently, $X_2$ is a mean-preserving spread of $X_1$ if there exists a random variable $Z$ such that for all $x_1\in\supp(X_1)$, we have (1) $\E{Z \mid X_1 = x_1}$, and (2) $X_1 + Z$ is equal in distribution to $X_2$.
    Note that we define mean-preserving spreads over the slightly more-general space of $\R^m$, and not just the space of posterior distributions (as is needed to state \autoref{thm:blackwell}); 
    this generality will be useful in \autoref{sec:triangle-spread} when we discuss mean-preserving spreads over the space of signaling schemes $\mathcal{P}$.
  }
\end{definition}

Formally, the equivalences we use can now be states as follows.

\begin{theorem}[\cite{blackwell1953equivalent}]
\label{thm:blackwell}
  Let $P : \Theta \to \Delta(\Omega)$ and $Q : \Theta \to \Delta(\Xi)$ be two signaling schemes with state space $\Theta$ and realization spaces $\Omega$ and $\Xi$.
  The following statements are equivalent:
\begin{enumerate}
    \item[\textbf{(i)}] For every decision problem with state space $\Theta$, any  Bayesian decision maker can achieve weakly higher expected utility under $P$ than under $Q$.\footnote{
        Formally, a decision problem specifies a prior over $\Theta$, a set of actions $A$, and a utility function $u : \Theta \times A \to \R$. 
        Under the signaling scheme $P$, the model is as follows.
        Nature draws the state $\theta$, the decision maker sees a signal realization drawn from $P(\theta)$, and then takes an action that maximizes her utility based on  the signal realization she receives.
    }
    \item[\textbf{(ii)}] $Q$ is a garbling of $P$.
    \item[\textbf{(iii)}] The posterior distribution $\tau_P$ is a mean-preserving spread of $\tau_Q$.
\end{enumerate}
In this case, we say that $P$ \textbf{Blackwell dominates} $Q$, denoted by $P \succeq Q$.
\end{theorem}

\subsection{The Dominance Ratio}
\label{sec:def-renyi-dominance-ratio}
\label{sec:prelims-dominance-ratio}

A Blackwell ordering relation $P\succeq Q$ captures a canonical and strong notion of $P$ being more informative than $Q$; 
however, it remains an ordinal ranking.
In other words, the Blackwell order can tell us which signaling scheme is better, but cannot tell us \emph{by how much}.

To address the above, a line of work has considered quantitative comparisons between signaling schemes $P, Q$ when learning from a large number of repeated independent samples from $P(\theta)$ vs. $Q(\theta)$. 
We follow the approach given by the \emph{dominance ratio} in \cite{mu2021blackwell}.\footnote{
  The main result of \cite{mu2021blackwell} characterizes the \emph{large sample} ordering of signaling schemes, which $P$ dominating $Q$ in large samples means that there is an $n$ large enough so that $n$ repeated draws from $P$ Blackwell dominates $n$ repeated draws from $Q$.
  This ordering was first introduces by \citet{azrieli2014comment} as a refinement of the ordering proposed by \cite{moscarini2002law} (where the above $n$ is allowed to depend on the particular decision problem).
  \citet{mu2021blackwell} prove \autoref{thm:mu-etal-dominanceratio} as a corrolary of their main characterization.
  In addition to providing a quantification of the advantage of a signal over another, this notion is able to compare all pairs of signaling schemes, and not only those that are Blackwell ranked.
}
The dominance ratio bounds the comparative number of samples from one signaling scheme needed to Blackwell dominate another signaling scheme.
For example, if observing 25 independent samples from a signaling scheme $P$ is Blackwell more informative than 100 independent samples from a signaling scheme $Q$, then intuitively, we say $P$ is (at least) 4 times as informative as $Q$ for large enough samples, and the dominance ratio is (at least) $4$.
This notion is formalized as follows.

\begin{definition}[\cite{mu2021blackwell}]
    The dominance ratio of two signaling schemes $P$ and $Q$ is defined as:

    \begin{equation*}
        P/Q = \sup \left\{\frac{m}{n}: P^{\otimes n} \succeq Q^{\otimes m} \right\} 
    \end{equation*}
    where $P^{\otimes n}$ denote the $n$-fold signaling scheme where $n$ independent observations are generated according to the signaling scheme $P$.
    \end{definition}
    
Intuitively, a dominance ratio of $r$ suggests that $P$ will be at least $r$ times as informative as $Q$ in large samples. 
We use the characterization from \cite{mu2021blackwell} of the dominance ratio in terms of the statistical notion of R\'enyi divergences.

\begin{definition}[\cite{renyi1961measures}]
    Given a signaling scheme $P$ with binary state space $\Theta = \{0,1\}$, signal space $\Omega$ and probability measures $P_0$ and $P_1$. The R\'enyi t-divergence of $P$ under $\theta$ is 
    \begin{align*}
        R_P^{\theta}(t) = R_t(P_{\theta} \parallel P_{1-\theta} ) = \frac{1}{t-1} \log \left( \int_{\Omega} \left( \frac{\mathrm{d}P_{\theta} }{\mathrm{d}P_{1-\theta} }(\omega) \right)^{t-1} \mathrm{d}P_{\theta}  \right)
    \end{align*}
    When $t \to 1$, it is the Kullback-Leibler divergence.
    \begin{align*}
        R_P^{\theta}(1) = R_1(P_{\theta} \parallel P_{1-\theta} ) =  \int_{\Omega}\log  \left( \frac{\mathrm{d}P_{\theta} }{\mathrm{d}P_{1-\theta} }(\omega) \right)\mathrm{d}P_{\theta} 
    \end{align*}
\end{definition}

\begin{theorem}[\cite{mu2021blackwell}]
\label{thm:mu-etal-dominanceratio}
Under binary states $\Theta = \{0,1\}$, we have
\begin{align*}
    P/Q = \inf_{\theta\in\{0,1\}, t>0} \frac{R_P^{\theta}(t)}{R_Q^{\theta}(t)}
\end{align*}
\end{theorem}

\section{Source-Aware and Source-Blind Learning: Preliminaries}
\label{sec:model}

We now return to our model of learning from a mixture of information sources, both with and without the knowledge of the source, and provide some preliminary analysis.

To begin our analysis, we consider the source-blind signaling scheme.
While one may worry that a source-blind decision maker may need to think in a sophisticated way about the distribution of signaling schemes, it turns out that there is a simple way to view source-blind learning. 
Namely, the blind signaling scheme $B(\Pi)$ exactly corresponds to the average of all signaling schemes in $\Pi$, where the average is calculated by treating $\pi$ as an element of 
$\R^{|\Theta| \times |\Omega|}$, given by $( p_{\theta\omega} )_{\theta \in \Theta, \omega\in\Omega}$.

\begin{proposition}\label{thm:source-blind}
     For any distribution of signaling schemes $\Pi \in \Delta(\mathcal{P})$, we have that $B(\Pi)$ is equivalent to the ``mean signal'' $\overline{\pi} = \E{\Pi}$.
\end{proposition}

Next, we show that for a given $\Pi$, source-aware decision makers always learn more information than source-blind ones.
This holds because the source-aware decision maker has the option to simply \emph{ignore the source}.
Formally:

\begin{proposition}\label{thm:AdominatesB}
    For any distribution of signaling schemes $\Pi \in \Delta(\mathcal{P})$, we have $A(\Pi) \succeq B(\Pi)$.
\end{proposition}

We end this section by introducing some notation.  In our analysis, it is sometimes useful to consider the distribution of $\Pi$ conditional on a given posterior $q \in S_k$ resulting from $A(\Pi)$.
We denote this by $\Pi|q$.  In other words, $\Pi|q$ denotes a source-aware decision maker's posterior distribution over $\Pi$, conditional on having received a signal realization such that their posterior distribution over $\Theta$ is $q$.

\autoref{sec:binary-example} provides an example illustrating \autoref{thm:source-blind} and \autoref{thm:AdominatesB} in the binary state and signal setting. Under this setting, we also provide a visualization for the distribution $\Pi|q$ in \autoref{fig:pi-conditional-on-q}.

\section{Source-Aware Learners Are Risk-loving in Information}
\label{sec:triangle-spread}

In this section we derive a result about how source-aware decision makers compare different signaling scheme distributions.
We show that source-aware decision makers always prefer ``more spread out'' distributions over signaling schemes.
This is formalized with the standard notion of a mean-preserving spread in the space of distributions over $\mathcal{P}$, where $\mathcal{P}$ is treated as a subset of multi-dimensional Euclidean space, via precisely the same definition as in \autoref{sec:prelims-blackwell}.
This will hold due to much the same intuition as in \autoref{thm:source-blind}: a source-aware learner will always prefer ``more spread out'' signaling schemes in the same way that a learner would rather \emph{be} source-aware, instead of learning based on only the mean of the distribution of signaling schemes, i.e., being source-blind.

The main result of this section is the following.

\begin{theorem}
\label{thm:triangle-spread-to-posterior-spread}
    For $\Pi_s, \Pi \in \Delta(\mathcal{P})$, suppose $\Pi_s$ is a mean-preserving spread of $\Pi$, as in \autoref{def:mps}.
    Then, $A(\Pi_s)$ Blackwell dominates $A(\Pi)$.
\end{theorem}

\section{A Special Case: Binary States and Signals}
\label{sec:binary}

All of our results to this point have applied to a general framework with finitely many states and signal realizations.  We now instantiate this general framework in the special case with binary states and 
binary signaling schemes, i.e., suppose we have $\Theta = \{0, 1\}$ and $\Omega = \{0,1\}$.
In this case, we denote a signaling scheme $\pi$ as $(p_{00}, p_{11})$, where $p_{00}$ and $p_{11}$ are the probability of the signal realization matching the state when the state is $0$ and $1$, respectively.
\autoref{tab:p0p1} illustrates such a signaling scheme.

We often further restrict attention to the domain of signaling schemes such that signal realization $0$ (resp., $1$) is always evidence that the state is $0$ (resp., $1$).
Formally, we denote the domain as $\mathcal{P}_{\urtriangle}
= \{(p_{00},p_{11}) \mid p_{00}+p_{11} \geq 1\}$.\footnote{
  A straightforward application of Bayes rule show that a signaling scheme $\pi$ increases the posterior belief that the state is $1$ when the signal realization is $1$ if and only if $\pi = (p_{00}, p_{11})$ has $p_{00}+p_{11}\ge 1$
}
We interpret this to mean that there is ``no fake data.''
Hence, we call this restriction the \emph{two-signal, no-fake-data} setting.
This domain of signaling schemes is illustrated in \autoref{fig:triangle}.

Given a signaling scheme $\pi$, recall that $\tau_{\pi} \in \Delta(\Delta(\Theta))$ denotes the distribution over posterior beliefs which arises from the signaling scheme $\pi$. 
Since we consider binary states, we typically represent $\tau_{\pi}$ as an element of $\Delta([0,1])$, where the event $\tau_{\pi} = q$ represents the posterior belief that the state is $1$ with probability $q$. 
By Bayes theorem, the decision maker with a uniform prior who knows $\pi$ and observes $\omega$ forms the posterior $p_{1\omega} / ( p_{0\omega} + p_{1\omega})$.

\def\arrvline{\hfil\kern\arraycolsep\vline\kern-\arraycolsep\hfilneg}

\begin{minipage}[c]{\textwidth}
  \begin{minipage}{0.5\textwidth}
    \begin{center}
    \begin{tabular}{cccccccc}
    &   & \multicolumn{2}{c}{Signal space $\Omega$} \\
    &   & $0$ & $1$
    \\ \cline{3-4}
    \multirow{2}{*}{\makecell{\vspace{-0.1in} \\ State space \\ $\Theta$}} 
      & $0$ \ \hspace{0.02in}\arrvline\hspace{-0.02in}
      &  $p_{00}$ 
      & \makecell{$p_{01}$ \\ {\footnotesize $= 1-p_{00}$} }
    \\[0.1in] 
      & $1$ \ \hspace{0.02in}\arrvline\hspace{-0.02in}
      & \makecell{$p_{10}$ \\ {\footnotesize $= 1-p_{11}$} }
      &  $p_{11}$ 
    \end{tabular}

    {\centering
    \captionof{table}{Binary signaling scheme specified by $(p_{00},p_{11})$} }
    \label{tab:p0p1}
    \end{center}
  \end{minipage}
  \begin{minipage}{0.5\textwidth}
  \centering
  \begin{tikzpicture}[scale=3]
            \draw[->] (0,0) -- (1.2,0) node[right] {$p_{00}$};
            \draw[->] (0,0) -- (0,1.2) node[above] {$p_{11}$};
        
            \draw[thick] (1,0) -- (1,1);
            \draw[thick] (0,1) -- (1,1);
            \draw[thick] (0,1) -- (1,0);
            
            \fill[gray!30] (0,1) -- (1,0) -- (1,1) -- cycle;
        
            \node[below left] at (0,0) {0};
            \node[below] at (1,0) {1};
            \node[left] at (0,1) {1};

        \end{tikzpicture}   
        
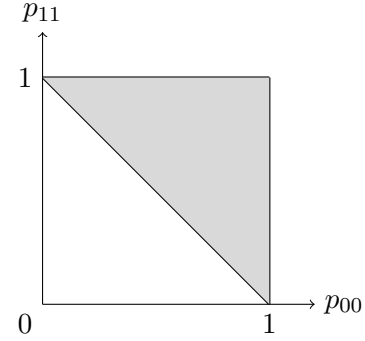
\captionof{figure}{Two-signal, ``no fake data'' domain}
        \label{fig:triangle}
\end{minipage}
\end{minipage}
\vspace{1em}

\autoref{tab:p0p1-simple-example} shows an example of a (no-fake-data) binary signaling scheme $\pi_z$ with $p_{00}=1/4$ and $p_{11}=7/8$.
The top mini-figure in \autoref{fig:posterior-simple-example} illustrates of the prior: with probability (abbreviated ``w.p.'') equal to $1$, the decision maker begins with the  belief that the state is $1$ with probability equal to $1/2$ (with this belief represented as the horizontal axis).
The bottom mini-figure illustrates the distribution of posteriors the decision maker will have after observing a signal realization from $\pi_z$: when the realization is $0$ (which occurs w.p. $3/16$), her posterior is $3/16$; when the realization is $1$ (which occurs w.p. $13/16$), her posterior is $7/13$.

\vspace{1em}
  \begin{minipage}{0.5\textwidth}
    \begin{center}
    \begin{tabular}{cccccccc}
    &   & \multicolumn{2}{c}{Signal space $\Omega$} \\
    &   & $0$ & $1$
    \\ \cline{3-4}
    \multirow{2}{*}{\makecell{\vspace{-0.1in} \\ State space \\ $\Theta$}} 
      & $0$ \ \hspace{0.02in}\arrvline\hspace{-0.02in}
      &  $1/4$  & $3/4$
    \\ 
      & $1$ \ \hspace{0.02in}\arrvline\hspace{-0.02in}
      & $1/8$ &  $7/8$
    \end{tabular}
    {\centering
    \captionof{table}{Binary signaling scheme specified by $(1/4,7/8)$}\label{tab:p0p1-simple-example} }
    \end{center}
  \end{minipage}
  \begin{minipage}{0.5\textwidth}
  \centering
      \includegraphics[width=0.7\textwidth]{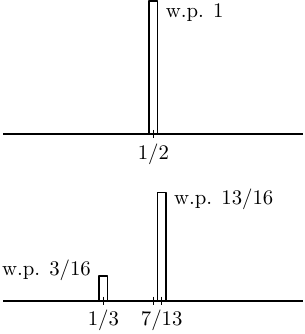}
      \captionof{figure}{Illustration of the prior, and of the distribution of posteriors generated by the signaling scheme in \autoref{tab:p0p1-simple-example}\label{fig:posterior-simple-example}}
\end{minipage}

\subsection{Source-Aware and Source-Blind Learning in Binary Settings}\label{sec:binary-example}

We next present an example illustrating \autoref{thm:source-blind} and \autoref{thm:AdominatesB} in the context of the binary setting.
Consider the case of binary signals under the ``no fake data'' assumption, where we represent signaling scheme as the pair $(p_{00}, p_{11})$, and let $\Pi$ be a uniform distribution over $\{ (1,0), (0,1), (1,1) \}$.
\autoref{thm:source-blind} says that $B(\Pi)$ is equivalent to the single signaling scheme $(2/3, 2/3)$; hence, $\tau_{B(\Pi)}$ is the uniform distribution on the posteriors $\{ 1/3, 2/3 \}$.
On the other hand, $A(\Pi)$ is equivalent to the signalling scheme that reveals the state with probability $1/3$ (and hence yielding posterior $0$ or $1$ with total probability $1/6$ each), and is completely uninformative with probability $2/3$.
\autoref{thm:AdominatesB} says that $A(\Pi)\succeq B(\Pi)$.
Thus, $\tau_{A(\Pi)}$ is a mean-preserving spread of $\tau_{B(\Pi)}$;
to see this directly from the distribution of posteriors, consider the random variable $Z$ such that:
\begin{align*}
& \P{Z=0 \mid \tau_{B(\Pi)}=1/3} = 1/3
& & 
\P{Z=1/2 \mid \tau_{B(\Pi)}=2/3} = 2/3
\\
& \P{Z=1/2 \mid \tau_{B(\Pi)}=1/3} = 2/3
& & 
\P{Z=1 \mid \tau_{B(\Pi)}=2/3} = 1/3.
\end{align*} 
Then, one can check that $Z$ satisfies the conditions of witnessing a mean-preserving spread as in \autoref{thm:blackwell}.
\autoref{fig:posterior-dist-a-b-example} illustrates this example of $B(\Pi)$ and $A(\Pi)$.

\paragraph{Posterior Distribution over Signal Sources: The Binary Case}

The binary model likewise makes it easier to visualize the distribution $\Pi|q$ over signal sources conditional on $A(\Pi)$.  Consider a full support $\Pi$ in the case of binary states and signals under the ``no fake data'' assumption, and recall that a posterior $q$ in this setting can be thought of as lying in $[0,1]$.
For $q>1/2$, the support of $\Pi|q$ consists of those points $(p_{00},p_{11})$ that satisfy $ \frac{p_{11}}{p_{11}+1-p_{00}} = q$;
one can show that this is
the line through $(1,0)$ and $(q,q)$ (not including the end point at $(1,0)$).
\autoref{fig:pi-conditional-on-q} illustrates this. 
Note, however, that $\Pi|q$ is not the distribution of $\Pi$ conditioned on lying in this line segment. Instead, $\Pi|q$ must weight each $\pi$ in this line segment by the probability that this $\pi$ induces posterior $q$ for the source-aware learner; for example, when $q>1/2$, distribution $\Pi|q$ must weight each $(p_{00}, p_{11})$ by $\frac12 (1 - p_{00}) + \frac12 p_{11}$.

\begin{minipage}[c]{1\textwidth}
  \begin{minipage}[c]{0.5\textwidth}
      \centering
      
      \vspace{0.2in}
      
      \includegraphics[width=0.7\textwidth]{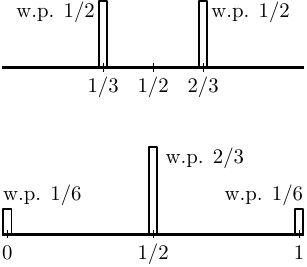}
      \captionof{figure}{Example of distribution of posteriors induced by $B(\Pi)$ and $A(\Pi)$}
      \label{fig:posterior-dist-a-b-example}
  \end{minipage}
  \begin{minipage}[c]{0.5\textwidth}
  \centering
  \begin{tikzpicture}[scale=3]
            \draw[->] (0,0) -- (1.2,0) node[right] {$p_{00}$};
            \draw[->] (0,0) -- (0,1.2) node[above] {$p_{11}$};
        
            \draw[thick] (1,0) -- (1,1);
            \draw[thick] (0,1) -- (1,1);
            \draw[thick] (0,1) -- (1,0);
            
            \fill[gray!30] (0,1) -- (1,0) -- (1,1) -- cycle;
        
            \node[below left] at (0,0) {0};
            \node[below] at (1,0) {1};
            \node[left] at (0,1) {1};

             \draw[thick] (1,0) -- (0.5,1);
         \node[right] at (0.7,0.65) {\textcolor{black}{$\Pi|q$}};
            \draw[thick] (1, 0) circle (0.6pt);

        \end{tikzpicture}   
        
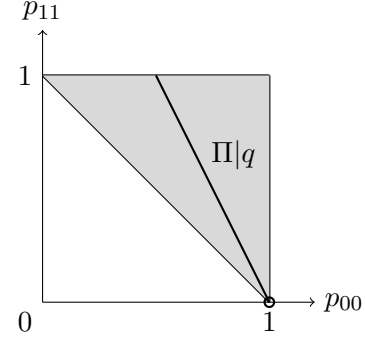
\captionof{figure}{$\Pi|q$ in binary signaling schemes}
        \label{fig:pi-conditional-on-q}
  \end{minipage}
\end{minipage}

\section{The Value of Many Sources in Binary Settings}\label{sec:value-of-sources}

\subsection{Over-provisioning Theorem}
\label{sec:over-provisioning}

\autoref{thm:AdominatesB} formalizes the fact that, for a fixed distribution of signaling schemes $\Pi$, source-aware learning, $A(\Pi)$, is more informative than source-blind learning, $B(\Pi)$.
However, our primary interest lies in quantifying \emph{how much} more informative $A(\Pi)$ is than $B(\Pi)$.
To formalize this, we make use of the dominance ratio $P/Q$ of \cite{mu2021blackwell} (recalled in \autoref{sec:def-renyi-dominance-ratio}).
We study $A(\Pi)/B(\Pi)$, representing ``how many times'' $A(\Pi)$ is more informative than $B(\Pi)$ in large samples.
We put an analytical upper bound on the value of $A(\Pi)/B(\Pi)$ in terms of the mean signal $\overline{\pi} = \E{\Pi}$.

In this section, we focus on binary states and the \emph{two-signal, no-fake-data} domain of signaling schemes. 
\citet{mu2021blackwell} restricts attention to binary states for technical reasons and defines dominance ratio in this environment.
Moreover, restricting to the domain $\mathcal{P}_{\urtriangle}$ gives us a natural approach to putting a meaningful upper bound on the value of sources.
In \autoref{sec:value-multiple-signal}, we discuss the challenges to upper bounding the value of sources with finite signal realizations or without the ``no fake data'' assumption. 

Our approach is to use \autoref{thm:triangle-spread-to-posterior-spread} to find the distribution of signaling schemes $\Pi$ on which source-aware learning, when compared to source-blind learning, has the largest dominance ratio. As stated in \autoref{thm:triangle-spread-to-posterior-spread}, $A(\Pi')$ Blackwell dominates $A(\Pi)$ as long as $\Pi'$ is a mean-preserving spread of $\Pi$. Since the domain of signaling schemes $\mathcal{P}_{\urtriangle}$ is convex, the distribution supported only on the vertices of the triangle is a mean-preserving spread to any distribution over $\mathcal{P}_{\urtriangle}$. This is formalized in the following lemma:

\begin{lemma}\label{thm:pi-star-unique}
    For any $\Pi \in \Delta(\mathcal{P}_{\urtriangle})$, there exists a unique $\Pi^*$ with support only on the vertices of $\mathcal{P}_{\urtriangle}$ such that $\Pi^*$ is a mean-preserving spread of $\Pi$.
\end{lemma}

\autoref{thm:pi-star-unique} establishes that for any $\Pi$, there is a unique $\Pi^*$ supported on the vertices that is a mean-preserving spread of $\Pi$.
Observe that a $\Pi^*$ is uniquely identified by its mean, which gives the weights on the vertices. 
Therefore, any  distribution $\Pi$ with the same mean signal $\E{\Pi} = \overline{\pi}$ will have the same $\Pi^*$ with  $\E{\Pi^*} = \overline{\pi}$ as their mean-preserving spread. 
With this observation, we provide an upper bound on the dominance ratio of any distribution $\Pi$ in terms of its mean signal $\overline{\pi} = \E{\Pi}$.

\begin{lemma}\label{thm:pi-star-ratio}
    For any $\overline{\pi} \in \mathcal{P}_{\urtriangle}$, let $\Pi^*$ denote the distribution supported on the vertices of $\mathcal{P}_{\urtriangle}$ with $\E{\Pi^*}=\overline{\pi}$.
    Then, for any other distribution $\Pi$ over $\mathcal{P}_{\urtriangle}$, we have 
    $A(\Pi)/B(\Pi)\le A(\Pi^*)/B(\Pi^*)$.
\end{lemma}

Now we are ready to state our main theorem, which provides an upper bound on the value of sources by analyzing the worst-case ratio obtained by $\Pi^*$.

\begin{theorem}[Over-provisioning]\label{thm:overprovisioning}
    Consider two-signal, no-fake-data signaling schemes. For any distribution of signaling schemes $\Pi \in \Delta(\mathcal{P}_{\urtriangle})$ with the average signal $\E{\Pi} = (x,y)$, $A(\Pi)/B(\Pi)$ is at most $ \frac{\log(2-x-y)}{\log(\sqrt{x(1-y)}+\sqrt{y(1-x)} )}$. 

    Furthermore, if the average signal is $\varepsilon$-away from being completely uninformative, i.e., $x+y \geq 1+\varepsilon$ for some $\varepsilon>0$, then $A(\Pi)/B(\Pi)$ is at most $ \frac{2\log(1-\varepsilon)}{\log (1-\varepsilon^2)} = O(1/\varepsilon)$.
\end{theorem}

We use numerical simulations to check the approximate tightness of the derived bound. 
Recall in \autoref{thm:mu-etal-dominanceratio}, the dominance ratio is the infiumum of a possibly non-convex function. Therefore, in \autoref{fig:dominance-ratio-simulation}, we numerically approximate the ratio for each $(x,y)$ pair by discretizing the domain into a grid, evaluating the function at those grid points, and taking the minimum value as an approximation of the infimum.
On the other hand, in \autoref{fig:dominance-ratio-analytical}, we directly evaluates the function $ \frac{\log(2-x-y)}{\log(\sqrt{x(1-y)}+\sqrt{y(1-x)} )}$ as in the theorem.
We plot the contour lines of values obtained by both methods. The analytical bound is reasonably close to the simulated infimum, providing informal evidence that this bound approximates the true value.

\vspace{1em}
 \begin{minipage}[c]{0.5\textwidth}
  \centering
\includegraphics[width=1\linewidth]{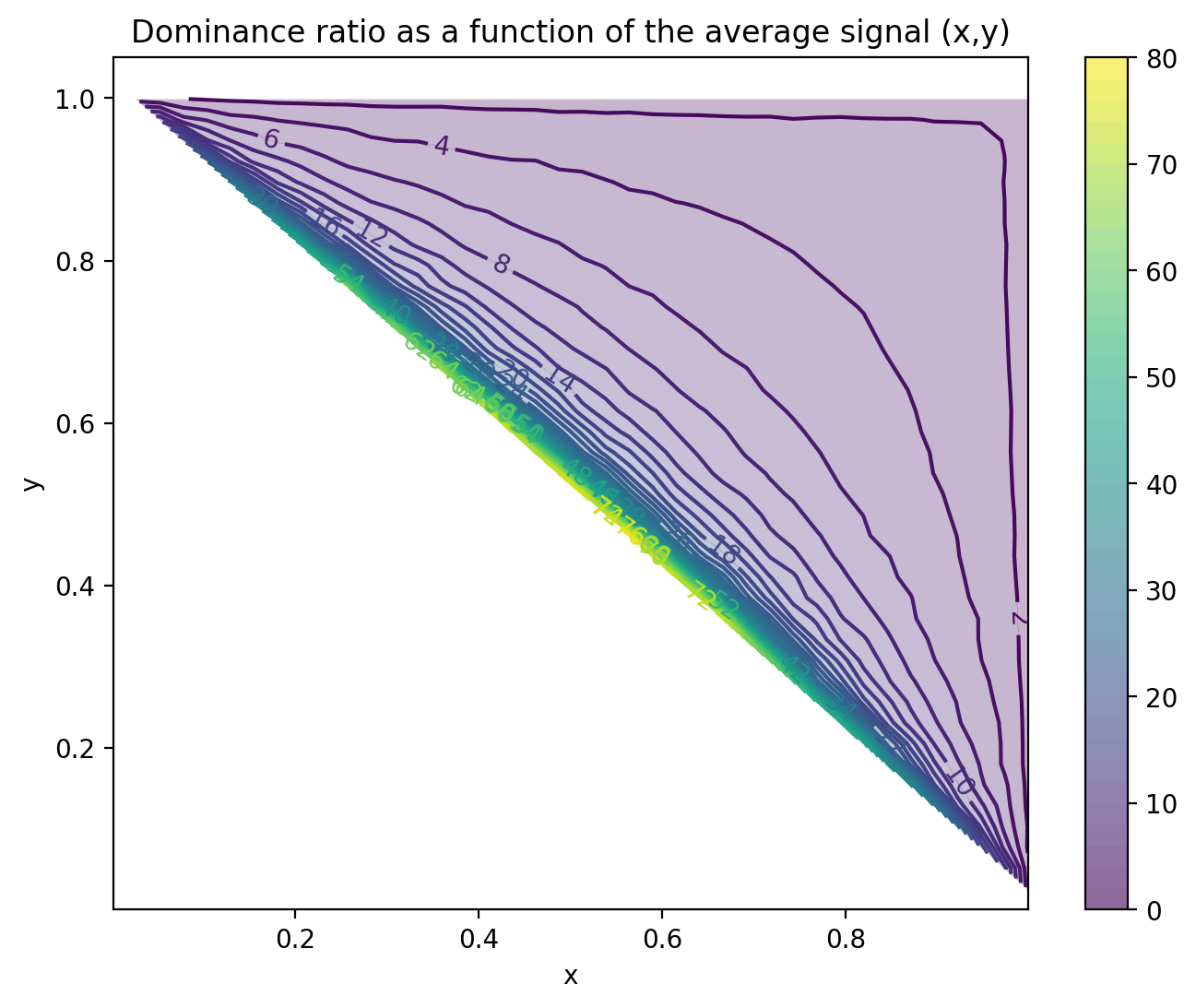}
  \captionof{figure}{Simulated dominance ratio}
  \label{fig:dominance-ratio-simulation}
\end{minipage}%
  \begin{minipage}[l]{0.5\textwidth}
  \centering
\includegraphics[width=1\linewidth]{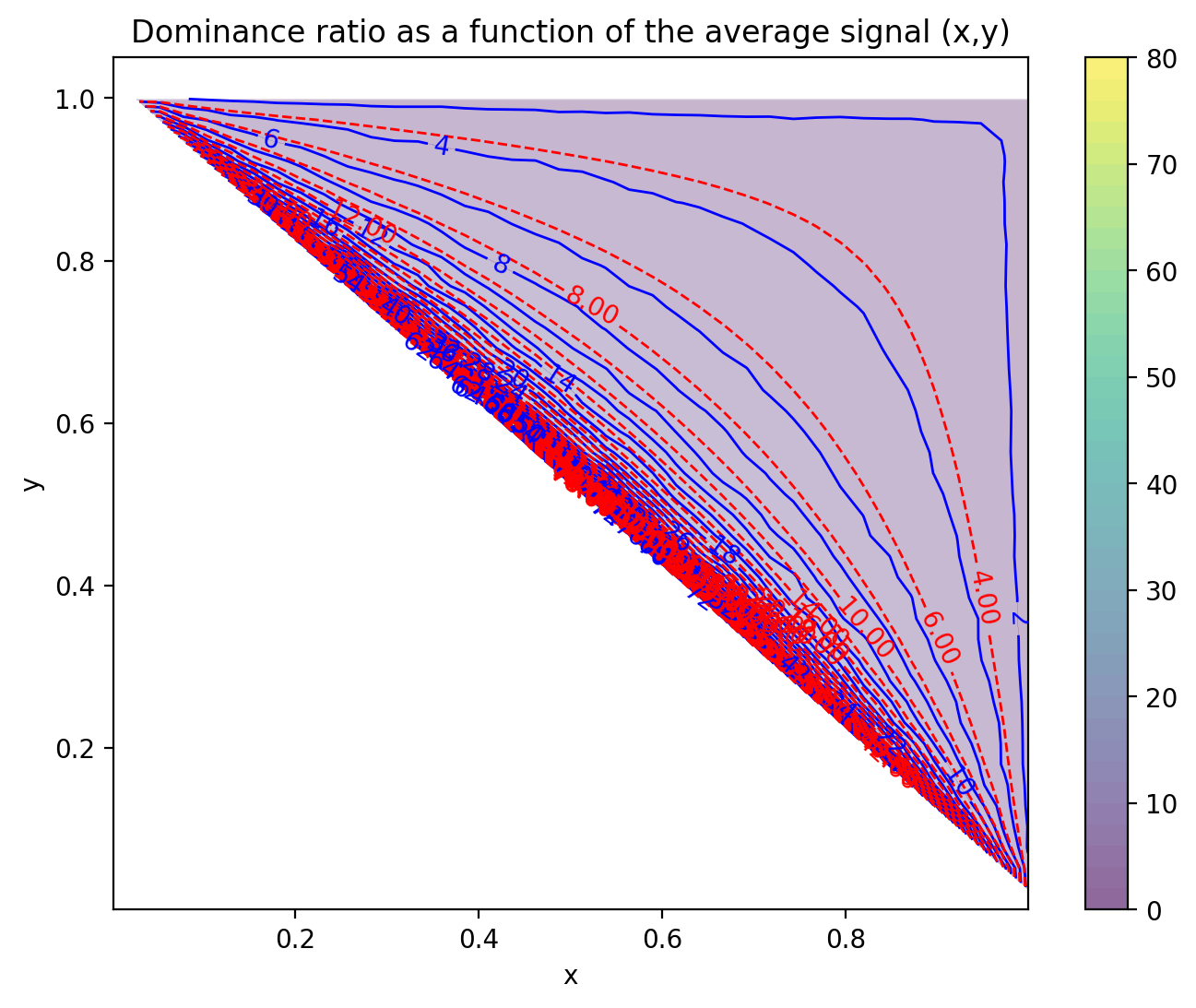}
     \captionof{figure}{Analytical upper bound}
     \label{fig:dominance-ratio-analytical}
\end{minipage}
\vspace{1em}

\begin{remark}
    The average signal quality is directly related to the maximal value one can obtain from knowing the sources. For distributions with an average signal that is relatively away from being uninformative, e.g. $(0.75,0.75)$, the maximal dominance ratio is $8$. On the other hand, for distributions with an average signal of $(0.51,0.51)$, the maximal dominance ratio is $200$. The intuition is that suppose an user only cares about the review made by one other person in the world, then it makes more sense to look for the source as opposed to read summarizations of more data. This corresponds to the situation where the average review is close to being uninformative to this user. 
\end{remark}

\section{Conclusion}\label{sec:conclusion}

In today’s information-rich environment, economic decision-making increasingly relies on aggregating large volumes of data from diverse sources. Whether evaluating product reviews, news articles, or medical recommendations, understanding who generates the information is crucial to judging its reliability and relevance. Yet, modern technologies that facilitate rapid data aggregation often obscure the origins of each data point, raising important questions about the trade-offs involved.

To address these concerns, we develop a formal framework that contrasts source-aware with source-blind learning. We extend the classic Blackwell model of learning by repeatedly drawing from a distribution of signaling schemes and either reveal (source-aware) or conceal (source-blind) the realized signaling scheme to the learner. 

Our analysis reveals source-aware learners benefit from increased heterogeneity among data sources whereas source-blind learning is mean-based. Focusing on a tractable binary model, our main result demonstrates that if the average signaling scheme is at least an $\varepsilon$-distance from being completely uninformative, then a source-blind learner requires at most 
$O(1/\varepsilon)$ times as much data to match the informativeness of a source-aware learner, as measured by Blackwell dominance in large samples. 

Our findings invite further exploration into more complex settings beyond the binary framework, as well as empirical investigations into real-world applications where the trade-off between data volume and source information is paramount.

\clearpage

\bibliographystyle{abbrvnat}
\bibliography{Bib}{}

\clearpage

\appendix

\section{Challenges in Extending to General Signaling Schemes}\label{sec:value-multiple-signal}

In \autoref{sec:binary}, our restriction to the \emph{two-signal, no-fake-data} signaling schemes provided a particularly natural path forward. 
Under both of these conditions, there is a unique mean-preserving spread of any arbitrary distribution $\Pi \in \Delta(\mathcal{P}_{\urtriangle})$ to the vertices of the space $\mathcal{P}_{\urtriangle}$.
This property let us define an $\Pi^*$ as the unique mean-preserving spread of $\Pi$ that is supported on the vertices. The R\'enyi divergences of $A(\Pi^*)$ and $B(\Pi^*)$ have good properties and can be upper bounded by a simple function form.

However, when we drop either one of the assumption, i.e., if we have signaling schemes with more than two realizations, or if the domain of binary signaling schemes has no restriction,
our approach runs into a problem: 
we lose the uniqueness of the mean-preserving spread supported on the vertices.
Hence, it is unclear how to approach getting a natural upper bound as \autoref{sec:over-provisioning} above.
We illustrate the challenge in a few examples.

First, consider the space of binary signaling schemes without the ``no-fake data'' assumption.
Observe that in this domain, there are distributions $\Pi$ such that $A(\Pi)$ is informative, but $B(\Pi)$ is completely uninformative.
For example, when $\Pi$ is a uniform distribution over the bottom left and top right corners $(0,0)$ and $(1,1)$, then $A(\Pi)$ always reveals the state, but $B(\Pi)$ is equivalent to $(1/2,1/2)$ and is completely uninformative. 
In particular, the dominance ratio $A(\Pi)/B(\Pi)$ is unbounded.

Moreover, even when we restrict attention to when the mean signal $\overline{\pi} = \E{\Pi}$ is bounded away from uninformativeness, 
there is no longer a unique mean-preserving spread from $\Pi$ to the vertices of the domain.
For example, the point $(3/4, 3/4)$ could be spread to a distribution $\Pi^*_a$ supported on vertices of $\mathcal{P}_{\urtriangle}$, as in \autoref{sec:over-provisioning}, or it could be spread to the distribution $\Pi^*_b$ that puts $1/4$ weight on $(0,0)$ and $3/4$ weight on $(1,1)$; the latter distribution $\Pi_b^*$ would (similar to the previous paragraph) have infinite dominance ratio and thus give no useful upper bound on the dominance ratio of $\Pi$.

Analogous problems arise when considering our problem for a larger number of states or signals.
To illustrate, consider the space of signaling schemes with two states and three signal realizations.
Since $|\Omega| = 3$ and $\mathcal{P} = \Delta(\Omega) \times \Delta(\Omega)$ is naturally viewed as a subspace of $\R^6$ (with dimension $4$). There are $9$ vertices of this polytope, specifically: 

\[
\begin{array}{ccc}
\pi_1 = \begin{bmatrix} 1 & 0 & 0 \\ 1 & 0 & 0 \end{bmatrix} &
\pi_2 = \begin{bmatrix} 0 & 1 & 0 \\ 1 & 0 & 0 \end{bmatrix} &
\pi_3 = \begin{bmatrix} 0 & 0 & 1 \\ 1 & 0 & 0 \end{bmatrix} \\[10pt]
\pi_4 = \begin{bmatrix} 1 & 0 & 0 \\ 0 & 1 & 0 \end{bmatrix} &
\pi_5 = \begin{bmatrix} 0 & 1 & 0 \\ 0 & 1 & 0 \end{bmatrix} &
\pi_6 = \begin{bmatrix} 0 & 0 & 1 \\ 0 & 1 & 0 \end{bmatrix} \\[10pt]
\pi_7 = \begin{bmatrix} 1 & 0 & 0 \\ 0 & 0 & 1 \end{bmatrix} &
\pi_8 = \begin{bmatrix} 0 & 1 & 0 \\ 0 & 0 & 1 \end{bmatrix} &
\pi_9 = \begin{bmatrix} 0 & 0 & 1 \\ 0 & 0 & 1 \end{bmatrix}
\end{array}
\]

Each vertices represents a signaling scheme where the columns denotes the space of signal realizations, the first row denotes $P_0 \in \Delta(\Omega)$, and the second row denotes $P_1 \in \Delta(\Omega)$. These vertices are either perfectly uninformative for all realizations (when $P_0 = P_1$), or perfectly informative (when $P_0\ne P_1$).
Consider a completely uninformative signaling scheme
\[\overline{\pi} = \begin{bmatrix} 1/3 & 1/3 & 1/3 \\ 1/3 & 1/3 & 1/3\end{bmatrix}. \] 
This signaling scheme can be written as $\frac13 \pi_1 + \frac13 \pi_5 + \frac13 \pi_9$, which are three uninformative signaling schemes. However, it can also be written as a $\frac19$ equal weighting of \emph{all} the vertices, which contains some informative signals.
More broadly, beyond two signals and states, it is unclear how to impose any sort of ``no-fake-data'' restriction which makes the above representations unique, as the signal realizations need not have any interpretation with regard to the states.

The above discussion shows that, without any restrictions on the space of signaling schemes, 
it is unclear how to get useful bounds on the dominance ratio due to non-uniqueness and possibly unbounded ratios. We leave further investigations in the space of general signaling schemes to future work.

\section{Omitted Proofs}\label{sec:omitted-proofs}

\subsection{Proof of \autoref{thm:source-blind}}
\begin{proof}
Observe that the distribution of signal realizations under $B(\Pi)$ exactly equals the distribution of signal realizations under $\overline{\pi}$ for each state.
Formally, in state $\theta$, signal realization $\omega$ occurs under $B(\Pi)$ with probability exactly
$ \Es{\pi\sim\Pi}{ p^\pi_{\theta\omega} }
= p^{\overline{\pi}}_{\theta\omega}$.
\end{proof}

\subsection{Proof of \autoref{thm:AdominatesB}}
\begin{proof}
    Recall that $A(\Pi)$ is a signaling scheme over $\mathcal P\times\Omega$, and $B(\Pi)$ is a signaling scheme over $\Omega$.
    Consider the function $\gamma: \mathcal{P} \times \Omega \to \Delta(\Omega)$ such that $\gamma(\pi,\omega) = \omega$.
    It holds directly by the definitions that $\gamma$ is a garbling from $A(\Pi)$ to $B(\Pi)$, as defined in \autoref{thm:blackwell}, and hence
    $A(\Pi)$ Blackwell dominates $B(\Pi)$.
\end{proof}

\subsection{Proof of \autoref{thm:triangle-spread-to-posterior-spread}}
\begin{proof}
Let $\Pi_s$ be a mean preserving spread of $\Pi$, and let $s : \R^m \to \Delta(\R^m)$ be the spread function in the space of signaling schemes as in \autoref{def:mps}.
Recall that this means that (1) for all $\pi\in\supp(\Pi)$, we have $\E{s(\pi)}=\pi$, and (2) if we sample $\pi\sim\Pi$ and then $\pi_s\sim s(\pi)$, then $\pi_s$ is equal in distribution to $\Pi_s$. 
By the equivalences in \autoref{thm:blackwell}, it suffices to show that  $\tau_{A(\Pi_s)}$ is \emph{also} a mean-preserving spread of $\tau_{A(\Pi)}$, in the space of posteriors $S_k$ (where recall that $S_k$ denotes the $k$-simplex).
We use $s(\cdot)$ to construct a spread function in the posterior space
$t : S_k \to \Delta( S_k )$ 
from $\tau_{A(\Pi)}$ to $\tau_{A(\Pi_s)}$.

Fix a certain $\pi$, consider $s(\pi)$, 
which is itself a distribution over signaling schemes.
Observe that for all $\pi$ in the support of $\Pi$, we have that $A(s(\pi))$ Blackwell dominates $\pi$, by the same logic we used to establish that $A(\Pi)$ always Blackwell dominates $B(\Pi)$ in \autoref{thm:AdominatesB}.\footnote{
  The crucial reliance of this step on \autoref{thm:source-blind} illustrates why mean-preserving spread over $\mathcal{P}$ are useful. 
  It also shows that our parametrization of the space of signaling schemes $\mathcal{P}$, and its notion of averages, may be the ``correct'' way to represent learning from a mixture of information sources.
}
By the equivalences in \autoref{thm:blackwell}, $\tau_{A(s(\pi))}$ is a mean-preserving spread of $\tau_{\pi}$. 
Thus, there is a spread function 
$t_\pi : S_k \to \Delta(S_k)$ 
from  $\tau_{\pi}$ to $\tau_{A(s(\pi))}$.
Now, for any posterior $q \in \supp(\tau_{A(\Pi)}) \subseteq S_k$, define $t(q)$ as follows: First, draw $\pi$ from $\Pi|q$; Second, draw $q_s$ from $t_\pi(q)$; Finally $t(q)$ is the distribution of $q_s$ that arises from the above process.\footnote{
  Note that when $\pi$ is drawn from $\Pi|q$, posterior $\tau_{\pi}$ must include $q$ in the support and $t_{\pi}$ must include $q$ in its domain. That's why we can always evaluate $t_{\pi}(q)$.
}
We will show that $t(\cdot)$ is a valid mean-preserving spread by showing: (1) it is mean-preserving at each point; and (2) the distribution after the spread equals the distribution $\tau_{A(\Pi_s)}$. 

(1) To see that $\E{ t(q) } = q$ for each $q$, we use the law of iterated expectation, and observe that for all $\pi$ in the support of $\Pi|q$, we know $t_\pi(q)$ is (by the definition of a spread function) a distribution with mean $q$. Therefore we can write:
\begin{equation*}
  \E{ t(q) }
  = \Es{\substack{q_s\sim t_{\pi}(q) \\ \pi \sim \Pi\mid q}}{q_s} 
  = \Es{\pi\sim \Pi\mid q}{ \Es{q_s\sim t_{\pi}(q)}{q_s | \pi}} 
  = \Es{\pi\sim \Pi\mid q}{q} 
  = q
\end{equation*}
(2) Now we must show that if $q_s$ is distributed by drawing $q$ from $\tau_{A(\Pi)}$ and then $q_s$ from $t(q)$, then $q_s$ is equal in distribution to $\tau_{A(\Pi_s)}$.
To see this, observe that by the definition of $t(q)$, this is equivalent to: 
First, draw $q$ from  $\tau_{A(\Pi)}$; Second, draw $\pi$ from $\Pi|q$; Third, draw $q_s$ from $t_\pi(q)$.
In this process, $\pi$ is distributed according to $\Pi$. 
Additionally, since $t_\pi$ is a spread function from $\tau_\pi$ to $\tau_{A(s(\pi))}$,
we know that for any fixed $\pi$, the distribution of $q_s$ conditioned on this $\pi$ equals the distribution of $\tau_{A(s(\pi))}$.
Since $s$ is a spread function from $\Pi$ to $\Pi_s$ as in \autoref{def:mps}, 
this means that if we draw $\pi\sim\Pi$ and then $\pi_s\sim s(\pi)$, then $\pi_s$ is equal in distribution to $\Pi_s$.
Therefore, the overall distribution of $q_s$ follows $\tau_{A(\Pi_s)}$, as desired.
This finishes the proof.
\end{proof}

\subsection{Proof of \autoref{thm:pi-star-unique}}
\begin{proof}
  Consider a spread function $s$ such that, for each point $(x,y)\in \mathcal{P}_{\urtriangle} \subseteq \R^2$, 
  the distribution $s(x,y)$ is supported only on the vertices $(1,1)$, $(0,1)$ and $(1,0)$.
  For any $(x,y)$, it is straightforward to see that there exists a unique distribution $s(x,y)$ on these vertices such that $\E{s(x,y)}=(x,y)$.
In particular, $s(x,y)$ is the distribution over $(1,1)$, $(0,1)$ and $(1,0)$ with probability $(x+y-1)$, $(1-x)$, and $(1-y)$, respectively. 
  Hence, distribution $\Pi^*$ arising from applying $s(\cdot)$ is unique as well. 
\end{proof}

\subsection{Proof of \autoref{thm:pi-star-ratio}}
\begin{proof}
    Recall the definition of the dominance ratio of $A(\Pi)$ over $B(\Pi)$:
    \begin{equation*}
        A(\Pi)/B(\Pi) = \sup \left\{\frac{m}{n}: A(\Pi)^{\otimes n} \succeq  B(\Pi)^{\otimes m} \right\} 
    \end{equation*}
  Since \autoref{thm:pi-star-unique} establishes that $\Pi^*$ is the unique mean-preserving spread of any $\Pi$ with the same mean, \autoref{thm:triangle-spread-to-posterior-spread} implies that $A(\Pi^*)$ Blackwell dominates $A(\Pi)$. This clearly implies that $A(\Pi^*)^{\otimes n}$ Blackwell dominates $A(\Pi)^{\otimes n}$ when repeatedly drawing from the same experiment $n$ times. Furthermore, we know $B(\Pi) = B(\Pi^*)$ from \autoref{thm:source-blind} since the distributions have the same mean. Now, for any $m,n$ such that $A(\Pi)^{\otimes n} \succeq B(\Pi)^{\otimes m}$ holds, we have that 
  \[A(\Pi^*)^{\otimes n} \succeq A(\Pi)^{\otimes n} \succeq B(\Pi)^{\otimes m} = B(\Pi^*)^{\otimes m}\]
  i.e., $A(\Pi^*)^{\otimes n} \succeq B(\Pi^*)^{\otimes m}$ also holds. We thus have that $A(\Pi^*)/B(\Pi^*) \geq A(\Pi)/B(\Pi)$ by definition.
\end{proof}

\subsection{Proof of \autoref{thm:overprovisioning}}
\begin{proof}
We use \autoref{thm:pi-star-ratio} and consider the dominance ratio for distributions supported on the vertices $\Pi^*$ as an upper bound. We also use the fact that the infimum of the R\'enyi ratio over $t>0$ is weakly smaller than the ratio evaluated at any $t$. Specifically, we use $t = \frac12$ as it simplifies the expression while providing a reasonably tight upper bound.

\[\frac{A(\Pi)}{B(\Pi)} \leq \frac{A(\Pi^*)}{B(\Pi^*)} = \inf_{\theta\in\{0,1\}, t>0} \frac{R_{A(\Pi^*)}^{\theta}(t)}{R_{B(\Pi^*)}^{\theta}(t)} \leq \frac{R_{A(\Pi^*)}^{\theta}(\frac12)}{R_{B(\Pi^*)}^{\theta}(\frac12)}\]

By \autoref{thm:source-blind}, $B(\Pi^*)$ is equivalent to the (binary) mean signal $(x,y)$. In binary signaling scheme, the R\'enyi divergences simplifies to:
\begin{align*}
    R_P^{\theta}(t) = \frac{1}{t-1} \log \left[\sum_{\omega=0}^{1}\left(\frac{p_{\theta\omega}}{p_{(1-\theta)\omega}}\right)^{t-1}\cdot p_{\theta\omega} \right] 
\end{align*}
When $t = \frac{1}{2}$, it further simplies to:
\begin{align*}
    R_P^{\theta}(1/2) &= -2 \log \left[\sum_{\omega=0}^{1} \sqrt{p_{\theta\omega} p_{(1-\theta)\omega}} \right] 
\end{align*}
Substituting $p_{00} = x$ and $p_{11} = y$:
\begin{align*}
    R_{(x,y )}^{\theta=0}(1/2) = R_{(x,y )}^{\theta=1}(1/2) = -2 \log \left[\sqrt{x(1-y)} + \sqrt{(1-x)y}\right]
\end{align*}

For $A(\Pi^*)$, on the other hand, the enriched signal space of $\mathcal{P}_{\urtriangle} \times \Omega$ is supported on
$\bigl\{ (1,0),\allowbreak (0,1),\allowbreak (1,1) \bigr\}\times \{0,1\}.$
The probability measure over the enriched signal space under state $\theta$ is 
\begin{align*}
    P_{\theta}(p,q,\omega) = f_{\Pi^*}(p,q)\cdot p_{\theta\omega}
\end{align*}
for $(p,q) \in \bigl\{ (1,0),\allowbreak (0,1),\allowbreak (1,1) \bigr\}$ and $\theta, \omega \in \{0,1\}$, where $f_{\Pi^*}$ is the probability mass function supported on the vertices $(1,0)$, $(0,1)$, and $(1,1)$.
The R\'enyi divergence for $A(\Pi^*)$ is calculated as follows, where we sum over  $(p,q) \in \bigl\{ (1,0),\allowbreak (0,1),\allowbreak (1,1) \bigr\}$.
\begin{align*}
    R_{A(\Pi^*)}^{\theta=0}(t)
    &= \frac{1}{t-1} \log \left[\sum_{p,q} \left(\frac{p^t}{(1-q)^{t-1}} + \frac{(1-p)^t}{q^{t-1}}\right) f_{\Pi^*}(p,q)\right] \\
      R_{A(\Pi^*)}^{\theta=1}(t)
    &= \frac{1}{t-1} \log \left[\sum_{p,q} \left(\frac{(1-q)^t}{p^{t-1}} + \frac{q^t}{(1-p)^{t-1}}\right) f_{\Pi^*}(p,q)\right] 
\end{align*}
When $t=\frac12$, this simplifies to:
\begin{align*}
    R_{A(\Pi^*)}^{\theta=0}(1/2) = R_{A(\Pi^*)}^{\theta=1}(1/2)
    &= -2 \log \left[\sum_{p,q}  \left(\sqrt{p(1-q)} + \sqrt{(1-p)q}\right) f_{\Pi^*}(p,q)\right]
\end{align*}

When $\Pi^*$ has mean $(x,y)$, the weight on the supporting signaling schemes for $A(\Pi^*)$ are: $f_{\Pi^*}(1,1) = x+y-1$, $f_{\Pi^*}(0,1) = 1-x$, and $f_{\Pi^*}(1,0) = 1-y$. Substituting to the expression, we obtain:
\[R^{\theta}_{A(\Pi^*)}(1/2) = -2 \log \left[
\sqrt{0}f_{\Pi}(1,1) + \sqrt{1}f_{\Pi}(1,0) + \sqrt{1}f_{\Pi}(0,1)\right] = -2 \log(2-x-y)\]
Therefore, for any signal distribution $\Pi$ with mean signal $(x,y)\in \mathcal{P}_{\urtriangle}$, $A(\Pi)/B(\Pi)$ can be upper bounded by the ratio of $\Pi^*$ with the same mean:
\begin{equation}\label{eq:xy-bound}
    \frac{A(\Pi)}{B(\Pi)} 
\leq 
\frac{R_{A(\Pi^*)}^{\theta}(1/2)}{R_{B(\Pi^*)}^{\theta}(1/2)}
= 
\frac{\log(2-x-y)}{\log(\sqrt{x(1-y)}+\sqrt{y(1-x)} )}
\end{equation}

Furthermore, if we put restrictions on the mean of $\Pi$, e.g., requiring that $x+y\geq 1+\varepsilon$, we can get more interpretable bounds in terms of $\varepsilon$.

First, let us consider for the $(x,y)$ pairs on the line $x+y = 1+\varepsilon'$ for any $\varepsilon'$. We want to find which $(x,y)$ pair on the line attains the highest dominance ratio. 
Substituting $y= 1+\varepsilon'-x$ to \autoref{eq:xy-bound}, we obtain
\[\frac{A(\Pi)}{B(\Pi)} 
\leq  \frac{\log(1-\varepsilon')}{\log(\sqrt{x(x-\varepsilon')}+\sqrt{(1-x+\varepsilon')(1-x)} )}\]

To bound the right hand side, we show that when $\varepsilon'\in(0,1)$ and $x\in [\epsilon,1]$, the logarithm 
attains its maximum at $x=\frac{1+\varepsilon'}{2}$.
We start by showing that the argument inside the logarithm is log-concave in $x$.
Since every concave function that is nonnegative on its domain is log-concave, we need to verify the concavity of argument. The second order derivative of $\sqrt{x(x-\varepsilon')}$ is $\frac{-(\varepsilon')^2}{4(x^2-x\varepsilon')^{3/2}} < 0$ when $\varepsilon'\in(0,1)$ and $x\in [\varepsilon',1]$. By a symmetric argument, the second order derivative of $\sqrt{(1-x+\varepsilon')(1-x)}$ is also strictly negative in the domain. The sum of two strictly concave functions is strictly concave. Therefore, the denominator is strictly concave in $x$ when $x\in [\epsilon,1]$ and attains the maximum at $x = \frac{1+\varepsilon'}{2}$. 
Since the numerator $\log(1-\varepsilon') <0$, the fraction also attains maximum at $x=\frac{1+\varepsilon'}{2}$. This means that the dominance ratio for $(x,y)$ along the line $x+y=1+\varepsilon'$ is upper bounded by the dominance ratio at $x=y=\frac{1+\varepsilon'}{2}$.

Substituting in $x=y=\frac{1+\varepsilon'}{2}$ to \autoref{eq:xy-bound} we obtain an upper bound in terms of $\varepsilon'$:

\begin{equation}\label{eq:epsilon-prime-bound}
    \frac{A(\Pi)}{B(\Pi)} 
\leq 
\frac{R_{A(\Pi^*)}^{\theta}(1/2)}{R_{B(\Pi^*)}^{\theta}(1/2)}
= 
\frac{2\log(1-\varepsilon')}{\log(1-(\varepsilon')^2)}
\end{equation}
We can rewrite the function as $\frac{2}{1+g(\varepsilon')}$ where $g(\varepsilon') = \frac{\log(1+\varepsilon')}{\log(1-\varepsilon')}$. By quotient rule, the numerator of the first-order derivative of $g$ is $\frac{\log(1-\varepsilon')}{1+\varepsilon'} + \frac{\log(1+\varepsilon')}{1-\varepsilon'}$, which is strictly positive when $\varepsilon'\in(0,1)$. Since $\frac{2}{1+g(\varepsilon')}$ is decreasing in $g(\varepsilon')$, we have proved that \autoref{eq:epsilon-prime-bound} is decreasing in $\varepsilon'$ when $\varepsilon'\in(0,1)$. 

Therefore, among the $(x,y)$ pairs such that $x+y \geq 1+\varepsilon$, those with $x+y = 1+\varepsilon$ attains the largest upper bound. Replacing $\varepsilon'$ with $\varepsilon$, we have that for any distribution of signaling schemes $\Pi$ with mean signal $(x,y)$ such that $x+y \geq 1+\varepsilon$, 
\begin{equation}\label{eq:epsilon-bound}
    \frac{A(\Pi)}{B(\Pi)} 
\leq 
\frac{2\log(1-\varepsilon)}{\log(1-\varepsilon^2)} \approx \frac{2}{\varepsilon}
\end{equation}

Using the approximation $\log(1+x) \approx x$ when the absolute value of $x$ is small, this fraction is approximately $\frac{2}{\varepsilon}$ when $\varepsilon$ is close to zero.
\end{proof}

\end{document}